\documentclass[journal,twoside,web]{ieeecolor}

\usepackage{generic}
\usepackage{cite}
\usepackage{amsmath,amssymb,amsfonts}
\usepackage{algorithm}
\usepackage{algpseudocode}

\usepackage{listings}
\usepackage{graphicx}
\usepackage{textcomp}
\usepackage{url}
\usepackage{hyperref}
\usepackage{graphicx}
\usepackage{subfig}
\usepackage{epsfig} 
\usepackage{cite}
\usepackage{lcsys}

\newtheorem{theorem}{Theorem}[section]
\newtheorem{lemma}[theorem]{Lemma}
\newtheorem{proposition}[theorem]{Proposition}
\newtheorem{corollary}[theorem]{Corollary}

\newtheorem{remark}[theorem]{Remark}
\newtheorem{assumption}[theorem]{Assumption}

\usepackage{amsmath} 
\usepackage{amssymb}  
\usepackage{xcolor}
\usepackage{caption}
\usepackage{graphicx}
\usepackage{mathtools}
\usepackage{caption}
\usepackage{float}
\usepackage{tikz}
\usepackage{url}

\usepackage{pgfplots}
\usepackage{xcolor}
\usetikzlibrary{matrix,arrows,calc,positioning,shapes,decorations.pathreplacing}
\usepackage{graphicx}

\usepackage{amsmath} 

\usepackage{enumerate}
\usepackage[all,tips]{xy}
\SelectTips{cm}{11}

\usepackage{bm}

\newcommand{\R}{\mathbb{R}}

\usepackage{mathtools}
\usepackage{siunitx}

\begin{document}

\def\BibTeX{{\rm B\kern-.05em{\sc i\kern-.025em b}\kern-.08em
    T\kern-.1667em\lower.7ex\hbox{E}\kern-.125emX}}
\markboth{\journalname, VOL. XX, NO. XX, XXXX 2017}
{Author \MakeLowercase{\textit{et al.}}: Preparation of Papers for IEEE Control Systems Letters (August 2022)}

\title{Safety-Critical Control for Quadrotor UAVs via Decentralized Navigation Functions}

\author{Omayra Yago Nieto$^{1}$, Alexandre Anahory Simoes$^{2}$ and Leonardo Colombo$^{3}$
\thanks{$^1$ Omayra Yago Nieto is with Universidad Politécnica de Madrid, Spain. {\tt\small omayra.yago.nieto@alumnos.upm.es}}
\thanks{$^2$ A. Anahory Simoes is with the School of Science and Technology, IE University, Spain
        {\tt\small alexandre.anahory@ie.edu}}
\thanks{$^3$Centre for Automation and Robotics (CSIC-UPM), Ctra. M300 Campo Real, Km 0,200, Arganda del Rey - 28500 Madrid, Spain, {\tt\small leonardo.colombo @csic.es}.}
\thanks{The research leading to these results was supported in part by iRoboCity2030-CM, Robótica Inteligente para Ciudades Sostenibles (TEC-2024/TEC-62), funded by the Programas de Actividades I+D en Tecnologías en la Comunidad de Madrid. The authors also acknowledge financial support from Grants PID2022-137909NB-C21 funded by MCIN/AEI/10.13039/501100011033.}}

\maketitle
\thispagestyle{empty}

\begin{abstract}
We study safety-critical control for teams of quadrotor UAVs driven
by decentralized navigation functions under learned model
uncertainty. These functions generate fully actuated translational
reference forces, while quadrotors can only produce thrust along
their body-fixed vertical axes. We construct a thrust-attitude
implementation of the induced navigation forces and quantify its
error with respect to the fully actuated reference dynamics. An
aggregated robust HOCBF-QP safety filter minimally modifies the
nominal thrusts while guaranteeing pairwise collision avoidance
with high probability.
\end{abstract}


\section{Introduction}

Teams of quadrotor UAVs are increasingly used in coordinated tasks such
as inspection, coverage, surveillance, aerial cinematography, and
cooperative sensing. These applications require controllers that combine
goal-directed motion, collision avoidance, and robustness to aerodynamic
and modeling uncertainty. A central difficulty is that decentralized motion-planning methods often
generate ideal translational fields or forces, whereas a quadrotor can
only generate thrust along its body-fixed vertical axis. Thus, a
collision-aware reference force cannot be applied directly, but must be
implemented through thrust and attitude commands while preserving safety.

Navigation functions provide a systematic way to generate
goal-directed and collision-aware motion. Starting from classical
artificial-potential constructions~\cite{RimonKoditschek1992}, several
decentralized navigation approaches have been developed for multi-agent
systems~\cite{DimarogonasKyriakopoulos2005, VerginisXuDimarogonas2017}.
Recent work has extended this viewpoint to multi-agent systems evolving
on \(SE(3)\), yielding decentralized navigation functions that generate
fully actuated reference forces for collision-aware motion
on \(SE(3)^N\)~\cite{YagoNietoSimoesGiribetColombo2026DNFSE3Learning}.
For quadrotor UAVs, however, these reference forces are not admissible
inputs. They must first be converted into desired thrust directions,
thrust magnitudes, and attitude commands. This conversion introduces an
implementation error that must be quantified before safety-critical
constraints can be imposed on the actual underactuated dynamics.

Control barrier functions (CBFs) provide a complementary framework for
enforcing safety through forward invariance of a prescribed safe set.
CBF-based quadratic programs allow safety constraints to be combined
with nominal controllers in a minimally invasive way
~\cite{AmesTAC2017}, and high-order or exponential CBFs
extend this approach to constraints of relative degree greater than one 
~\cite{NguyenACC2016,XiaoCDC2019,TanShawCortezDimarogonas2022}.
Because collision avoidance for quadrotors is naturally a
relative-degree-two constraint with respect to thrust, HOCBFs are well
suited to this setting. Gaussian Processes (GP) have also been used to incorporate learning
uncertainty into safety-critical control, including quadrotor
learning, GP-CBF synthesis, high-order GP-CBFs, and decentralized
multirobot collision avoidance
~\cite{WangTheodorouEgerstedt2018,JagtapPappasZamani2020,AaliLiu2024,HuFuWen2023}. Related NMPC methods address safe multi-quadrotor operation under sensing and actuation constraints \cite{GoarinLiSavioloLoianno2025}.

Despite these advances, the combination of decentralized navigation
functions, quadrotor underactuation, learned model uncertainty, and
HOCBF-based safety remains largely unexplored. The main issue is that
the navigation function produces a fully actuated reference force, while
the HOCBF constraint must be enforced on the actual thrust-direction
dynamics. This paper bridges that gap. The contributions are threefold.
First, we construct a thrust--attitude implementation of decentralized
navigation forces for quadrotor teams and derive an explicit error bound
with respect to the fully actuated reference dynamics. Second, we
incorporate GP force uncertainty directly into the pairwise
relative-degree-two HOCBF constraint through an explicit robustness
bound. Third, we formulate an aggregated robust HOCBF-QP safety filter
that minimally modifies the nominal thrusts and guarantees
pairwise collision avoidance with high probability.

The remainder of the paper is organized as follows. Section~II
introduces the multi-quadrotor model and the pairwise HOCBF safety
constraints. Section~III develops the decentralized-navigation-force
implementation and derives the thrust--attitude error bound.
Section~IV constructs the aggregated robust GP-HOCBF-QP safety
filter and establishes the probabilistic safety guarantee. Section~V presents simulation results, and Section~VI concludes.

\section{System Modeling and Safety Constraints}
\label{sec:system_model}

\subsection{Multi-agent quadrotor dynamics}

Consider a team of $N$ quadrotor-type agents indexed by
$\mathcal V=\{1,\ldots,N\}$. The configuration of agent
$i\in\mathcal V$ is $g_i=(R_i,p_i)\in SO(3)\times\R^3$, where
$R_i\in SO(3)$ is the attitude and $p_i\in\mathbb R^3$ is
the position in an inertial frame. Let $v_i\in\mathbb R^3$
denote the translational velocity and let
$\Omega_i\in\mathbb R^3$ denote the body angular velocity.
The kinematics are
\begin{equation}
    \dot p_i=v_i,
    \qquad
    \dot R_i=R_i\widehat\Omega_i ,
    \label{eq:kinematics}
\end{equation}
where $\widehat{\cdot}:\mathbb R^3\to\mathfrak{so}(3)$ is
the standard hat map.

The control input is $u_i=(f_i,\tau_i)\in\mathbb R\times
\mathbb R^3$, where $f_i$ is the thrust magnitude along the
body-fixed direction $R_ie_3$ and $\tau_i$ is the body torque.
The dynamics are
\begin{align}
    m_i\dot v_i
    &=
    f_iR_ie_3+m_i g+\Delta_{f,i}(x_i),
    \label{eq:trans_dyn}
    \\
    J_i\dot\Omega_i+\Omega_i\times J_i\Omega_i
    &=
    \tau_i+\Delta_{\tau,i}(x_i),
    \label{eq:rot_dyn}
\end{align}
where $m_i>0$, $J_i\in\mathbb R^{3\times3}$,
$e_3=(0,0,1)^\top$, $g\in\mathbb R^3$ is the gravitational
acceleration vector and it is directed along $-e_{3}$, and
$ x_i=(p_i,v_i,R_i,\Omega_i)\in
    \mathbb R^3\times\mathbb R^3\times SO(3)\times\mathbb R^3$.

The terms $\Delta_{f,i}$ and $\Delta_{\tau,i}$ represent
state-dependent unknown force and torque contributions, respectively.
They account for unmodeled aerodynamic effects, parametric uncertainty,
actuator disturbances, and external perturbations. The nominal dynamics
are recovered by setting $\Delta_{f,i}=\Delta_{\tau,i}=0$. In this work,
these unknown dynamics are estimated online by learning-based methods and
their prediction errors are incorporated into the robust safety-critical
controller. The translational dynamics remain underactuated because the
control force cannot be chosen arbitrarily in $\mathbb R^3$, but is
constrained to lie in $\operatorname{span}\{R_i e_3\}$ (for details, see \cite{lee2010geometric}).
\subsection{Pairwise safety functions}

Collision-avoidance constraints are encoded by an undirected
graph $\mathcal G=(\mathcal V,\mathcal E)$. An edge
$(i,j)\in\mathcal E$ indicates that agents $i$ and $j$ must
satisfy a pairwise safety constraint. The neighbor set of agent
$i$ is $\mathcal N_i=\{j\in\mathcal V:(i,j)\in\mathcal E\}$.

For collision avoidance, each quadrotor is represented by a
spherical safety envelope centered at its position $p_i$.
Therefore, pairwise safety depends only on the relative
position between agents and not on their attitudes. For each
$(i,j)\in\mathcal E$, define
\begin{equation}
    h_{ij}(p_i,p_j)
    =
    \|p_i-p_j\|^2-d_{\min}^2,
    \label{eq:pairwise_safety_function}
\end{equation}
where $d_{\min}>0$ is the prescribed minimum separation
distance. The pairwise and global safe sets are
\begin{equation}
    \mathcal C_{ij}
    =
    \{(x_i,x_j):h_{ij}(p_i,p_j)\geq0\},
    \qquad
    \mathcal C
    =
    \bigcap_{(i,j)\in\mathcal E}\mathcal C_{ij}.
    \label{eq:safe_sets}
\end{equation}

Let $p_{ij}:=p_i-p_j$ and $v_{ij}:=v_i-v_j$. Along
\eqref{eq:kinematics}--\eqref{eq:rot_dyn},
\begin{equation}
    \dot h_{ij}
    =
    2p_{ij}^{\top}v_{ij}.
    \label{eq:h_dot}
\end{equation}
The relative acceleration entering the second derivative is
\begin{equation}
    a_{ij}
    :=
    \frac{f_i}{m_i}R_ie_3
    -
    \frac{f_j}{m_j}R_je_3
    +
    \frac{\Delta_{f,i}(x_i)}{m_i}
    -
    \frac{\Delta_{f,j}(x_j)}{m_j}.
    \label{eq:relative_acceleration}
\end{equation}
Since the common gravitational acceleration cancels from the
relative dynamics, we obtain
\begin{equation}
    \ddot h_{ij}
    =
    2\|v_{ij}\|^2
    +
    2p_{ij}^{\top}a_{ij}.
    \label{eq:h_ddot}
\end{equation}

Therefore, the control inputs do not appear in $\dot h_{ij}$,
whereas the thrust magnitudes $f_i$ and $f_j$ appear explicitly
in $\ddot h_{ij}$ through the underactuated force directions
$R_ie_3$ and $R_je_3$. Thus, except at configurations where \(p_{ij}\) is orthogonal to both
available thrust directions, \(h_{ij}\) has relative degree two with
respect to the thrust inputs. This motivates the use of high-order
control barrier functions for enforcing~\eqref{eq:pairwise_safety_function}
\cite{NguyenACC2016,XiaoCDC2019}. The body torques
$\tau_i,\tau_j$ do not enter $\ddot h_{ij}$ directly; they affect
safety indirectly by changing the attitudes $R_i,R_j$ and,
therefore, the future thrust directions. This is precisely the
quadrotor underactuation structure given in \cite{lee2010geometric}.

\subsection{High-order barrier constraints}

Since the pairwise safety function has relative degree two
with respect to the thrust inputs, we employ the high-order
zeroing-CBF framework to enforce~\eqref{eq:pairwise_safety_function}.
This framework extends barrier-certificate and zeroing-CBF
ideas for forward invariance to constraints whose control input
appears only after repeated differentiation
\cite{AmesTAC2017,
NguyenACC2016,XiaoCDC2019,TanShawCortezDimarogonas2022}.

Let $\alpha_1$ and
$\alpha_2$ be differentiable extended class-$\mathcal K$ functions and define
\begin{equation}
    \varphi_{0,ij}:=h_{ij},
    \qquad
    \varphi_{1,ij}:=\dot h_{ij}+\alpha_1(h_{ij}),
    \label{eq:hocbf_phi01}
\end{equation}
\begin{equation}
    \varphi_{2,ij}:=
    \dot\varphi_{1,ij}+\alpha_2(\varphi_{1,ij}).
    \label{eq:hocbf_phi2}
\end{equation}
Equivalently,
\begin{equation}
    \varphi_{2,ij}
    =
    \ddot h_{ij}
    +
    \alpha_1'(h_{ij})\dot h_{ij}
    +
    \alpha_2\!\left(\dot h_{ij}+\alpha_1(h_{ij})\right).
    \label{eq:hocbf_explicit}
\end{equation}
In the common linear choice
$\alpha_1(s)=\ell_1s$, $\alpha_2(s)=\ell_2s$,
\eqref{eq:hocbf_explicit} becomes
\[
    \ddot h_{ij}
    +
    (\ell_1+\ell_2)\dot h_{ij}
    +
    \ell_1\ell_2 h_{ij}
    \geq 0 .
\]
If \(\phi_{0,ij}(0)\ge0\) and \(\phi_{1,ij}(0)\ge0\), then
enforcing \(\phi_{2,ij}(x,u)\ge0\) guarantees forward invariance of the
high-order safe set
\[
\mathcal C^H_{ij}:=\{(x_i,x_j):\phi_{0,ij}\ge0,\ \phi_{1,ij}\ge0\}.
\]
Since \(\phi_{0,ij}=h_{ij}\), this implies \(h_{ij}(t)\ge0\) for all
\(t\ge0\)~\cite{NguyenACC2016,XiaoCDC2019,TanShawCortezDimarogonas2022}.
The robust enforcement of this HOCBF condition under model
uncertainty and underactuated thrust-direction constraints is
developed in Section~\ref{sec:robust_safety_filter}.

\section{Decentralized Navigation Functions and Thrust-Attitude Control}
\label{sec:navigation_control}

\subsection{Decentralized navigation functions}

For each agent $i\in\mathcal V$, let
$g_{d,i}=(R_{d,i},p_{d,i})\in SE(3)$ denote its desired
configuration. 
Choose \(d_{\rm nav}<d_{\min}\) and define
\[
\mathcal Q_{\rm nav}:=
\{g\in SE(3)^N:\|p_{ij}\|>d_{\rm nav},
\ \forall(i,j)\in\mathcal E\},
\]
where \(p_{ij}:=p_i-p_j\). \(\mathcal Q_{\rm nav}\) is a
configuration-space object, whereas the safe set \(\mathcal C\)
in~\eqref{eq:safe_sets} is defined on the state space. Since
\(d_{\rm nav}<d_{\min}\), the navigation function is defined on an
open neighborhood of the CBF safe set boundary
\(\|p_{ij}\|=d_{\min}\).

We use the decentralized navigation function on $SE(3)^N$
introduced in~\cite{YagoNietoSimoesGiribetColombo2026DNFSE3Learning}.
This construction extends classical navigation functions
\cite{RimonKoditschek1992} and earlier decentralized
multi-agent navigation methods
\cite{DimarogonasKyriakopoulos2005} to fully actuated
dynamical agents evolving on $SE(3)$. For agent $i$, this
function is denoted by
$\psi_i:\mathcal Q_{\rm nav}\rightarrow\mathbb R_{\geq0}$.
It is built from relation proximity functions, relation
verification functions, and a collision-anticipation correction.

\begin{assumption}
\label{ass:navigation_function}
For each $i\in\mathcal V$, the function
$\psi_i:\mathcal Q_{\rm nav}\rightarrow\mathbb R_{\geq0}$
satisfies:
\begin{enumerate}
    \item $\psi_i(g)\geq0$ for all $g\in\mathcal Q_{\rm nav}$,
    and $\psi_i(g)=0$ when $g_i=g_{d,i}$ and no active collision
relation involving agent $i$ is present;

    \item $\nabla_{p_i}\psi_i$ depends only on locally available
    configuration information associated with agent $i$ and
    the agents involved in its collision relations;

    \item the negative gradient $-\nabla_{p_i}\psi_i$ is
    attractive toward the desired configuration and repulsive
    with respect to collision configurations.
\end{enumerate}
\end{assumption}

A representative expression of the decentralized navigation
function is
\begin{equation}
    \psi_i(g_1,\ldots,g_N)
    =
    \frac{
        \gamma_{d,i}(g_i)+c_i(G_i)
    }{
        \left(
            \big(\gamma_{d,i}(g_i)+c_i(G_i)\big)^\kappa
            +
            G_i
        \right)^{1/\kappa}
    },
    \label{eq:decentralized_navigation_function}
\end{equation}
where $\kappa>0$ is a design parameter, $\gamma_{d,i}$ is a
goal function, $G_i$ collects the relation-verification
functions associated with possible collision relations involving
agent $i$, and $c_i(G_i)$ is the collision-anticipation
correction term. This correction prevents agent $i$ from
remaining passive at its desired configuration when a potential
collision relation involving that agent is still active. The
detailed construction of $G_i$ and $c_i$ is given
in~\cite{YagoNietoSimoesGiribetColombo2026DNFSE3Learning}.

The function $\psi_i$ is used as a decentralized artificial
potential on $\mathcal Q_{\rm nav}$. Its negative position
gradient determines the nominal collision-aware direction in
which agent $i$ should move. If the translational dynamics
were fully actuated, this direction could be imposed directly
as a feedback force. For a quadrotor UAV, however, the
available force is constrained to the body-fixed thrust
direction $R_ie_3$. Therefore, the control action induced by
$\psi_i$ must be implemented through a thrust magnitude and
a desired attitude before it can be applied to the dynamics.
\subsection{Fully actuated reference dynamics}

The decentralized navigation function induces the ideal
translational feedback force
\begin{equation}
    F_i^{\rm nav}(x)
    :=
    -k_p\nabla_{p_i}\psi_i(g)
    -
    k_v v_i,
    \label{eq:nav_force}
\end{equation}
where $k_p,k_v>0$ are navigation and damping gains,
respectively, and $g=(g_1,\ldots,g_N)$. The first term is the
negative position gradient of the decentralized navigation
function and therefore combines attraction to the desired
configuration with repulsion from collision configurations.
The second term introduces linear damping in the translational
velocity. For time-varying reference positions, the damping term can be
replaced by $-k_v(v_i-v_{d,i})$, with $v_{d,i}=\dot p_{d,i}$;
the thrust--attitude implementation and the results below are
unchanged.

If agent $i$ were fully actuated in translation, the reference
closed-loop dynamics associated with~\eqref{eq:nav_force}
would be
\begin{equation}
    m_i\dot v_i
    =
    F_i^{\rm nav}(x).
    \label{eq:fully_actuated_reference}
\end{equation}
Thus, $F_i^{\rm nav}$ should be interpreted as a virtual
fully actuated force generated by the decentralized navigation
function. This is precisely the type of force that can be
applied in the fully actuated $SE(3)$ model of
\cite{YagoNietoSimoesGiribetColombo2026DNFSE3Learning}.

For a quadrotor UAV, however, the admissible translational
force is not arbitrary. The actual dynamics satisfy
\[
    m_i\dot v_i
    =
    f_iR_ie_3+m_ig+\Delta_{f,i}(x_i),
\]
so the control force is constrained to the one-dimensional
subspace $\operatorname{span}\{R_ie_3\}$. Consequently,
the reference force~\eqref{eq:nav_force} cannot be applied
directly. It must first be converted into a desired thrust
direction, a thrust magnitude, and an attitude command. Therefore, the decentralized navigation function affects the
quadrotor dynamics only through the virtual force
$F_i^{\rm nav}$, which is implemented below by means of
thrust and attitude commands.

\subsection{Thrust-attitude implementation}

We compensate the unknown force and torque terms using the
learning-based framework for underactuated vehicles with
uncertain dynamics in~\cite{BeckersColomboHirchePappas2021CSL}. Each agent uses collected data to obtain mean predictions $\mu_{f,i}(x_i)$
and $\mu_{\tau,i}(x_i)$ of $\Delta_{f,i}(x_i)$ and
$\Delta_{\tau,i}(x_i)$, respectively. The prediction uncertainty is incorporated into the robust
HOCBF constraints in Section~\ref{sec:robust_safety_filter}. 

We define the desired total acceleration
\begin{equation}
    a_{d,i}(x)
    :=
    \frac{1}{m_i}
    \big(F_i^{\rm nav}(x)-\mu_{f,i}(x_i)\big)
    -
    g.
    \label{eq:desired_acceleration}
\end{equation}
Equivalently, if this acceleration were implemented exactly
and the learned force model were exact, then
\[
    m_i a_{d,i}+m_i g+\mu_{f,i}=F_i^{\rm nav}.
\]
Whenever \(a_{d,i}(x) \neq 0\), define
\(b_{3c,i}(x) := a_{d,i}(x)/\|a_{d,i}(x)\|\). Let
\(R_{c,i}\in SO(3)\) be any smooth commanded attitude satisfying
\(R_{c,i}e_3=b_{3c,i}\); the remaining degree of freedom may be
assigned through a desired yaw direction.

Since the actual thrust is applied along the current direction
$R_ie_3$, we choose
\begin{equation}
    f_i^{\rm nom}
    =
    m_i a_{d,i}^{\top}R_ie_3 .
    \label{eq:nominal_thrust}
\end{equation}
Thus, the applied force $f_i^{\rm nom}R_ie_3$ is the
orthogonal projection of $m_i a_{d,i}$ onto
$\operatorname{span}\{R_ie_3\}$.

The commanded attitude \(R_{c,i}\) is tracked with the geometric
controller of~\cite{lee2010geometric}. Let
\(\widehat\Omega_{c,i}=R_{c,i}^{\top}\dot R_{c,i}\), and let
\(e_{R,i},e_{\Omega,i}\) be the standard tracking errors. We use
\[
\begin{aligned}
\tau_i^{\rm nom}
={}&-k_Re_{R,i}-k_\Omega e_{\Omega,i}
+\Omega_i\times J_i\Omega_i\\
&-J_i\big(
\widehat\Omega_iR_i^\top R_{c,i}\Omega_{c,i}
-R_i^\top R_{c,i}\dot\Omega_{c,i}
\big)
-\mu_{\tau,i}(x_i).
\end{aligned}
\]
The corresponding nominal input is
\(u_i^{\rm nom}=(f_i^{\rm nom},\tau_i^{\rm nom})\).

The following result quantifies the error between the
fully actuated reference dynamics and their thrust-attitude
implementation on the quadrotor.

\begin{proposition}
\label{prop:implementation_error}
Consider agent $i$ with dynamics~\eqref{eq:trans_dyn}
under the nominal thrust~\eqref{eq:nominal_thrust}. Suppose
that $a_{d,i}(x)\neq0$ and
\begin{equation}
    \|R_ie_3-b_{3c,i}\|\leq\varepsilon_i .
    \label{eq:thrust_direction_error}
\end{equation}
Then
\begin{equation}
    m_i\dot v_i
    =
    F_i^{\rm nav}(x)
    +
    e_i^{\rm att}(x)
    +
    e_i^{\rm GP}(x_i),
    \label{eq:implemented_nav_dynamics}
\end{equation}
where
\begin{align}
    e_i^{\rm GP}(x_i)
    &:=
    \Delta_{f,i}(x_i)-\mu_{f,i}(x_i),
    \label{eq:gp_error}
    \\
    e_i^{\rm att}(x)
    &:=
    m_i
    \big[
        (a_{d,i}^{\top}R_ie_3)R_ie_3
        -
        a_{d,i}
    \big].
    \label{eq:att_error}
\end{align} 
Moreover, if \(\theta_i\in[0,\pi]\) is the angle between
\(R_ie_3\) and \(b_{3c,i}\), 
\begin{equation}
    \|e_i^{\rm att}(x)\|
    =
    m_i\|a_{d,i}(x)\|\sin\theta_i
    \leq
    m_i\|a_{d,i}(x)\|\varepsilon_i .
    \label{eq:att_error_bound}
\end{equation}
Consequently, if $e_i^{\rm GP}=0$ and $R_ie_3=b_{3c,i}$,
then the quadrotor translational dynamics coincide with the
fully actuated reference dynamics~\eqref{eq:fully_actuated_reference}.
\end{proposition}

\begin{proof}
Using~\eqref{eq:trans_dyn} and~\eqref{eq:nominal_thrust},
\[
m_i\dot v_i
=
m_i(a_{d,i}^{\top}R_ie_3)R_ie_3
+
m_ig
+
\Delta_{f,i}(x_i).
\]
Since \(m_i a_{d,i}+m_ig+\mu_{f,i}=F_i^{\rm nav}\), this gives
\eqref{eq:implemented_nav_dynamics} with
\eqref{eq:gp_error}--\eqref{eq:att_error}. Let \(u_i:=R_ie_3\) and \(b_i:=b_{3c,i}\). Since
\(a_{d,i}=\|a_{d,i}\|b_i\),
$e_i^{\rm att}
=
m_i\|a_{d,i}\|\big((b_i^\top u_i)u_i-b_i\big)$. The vector in parentheses is the component of \(b_i\) orthogonal to
\(u_i\), up to sign. Hence, if \(\theta_i\) is the angle between
\(u_i\) and \(b_i\),
$\|e_i^{\rm att}\|
=
m_i\|a_{d,i}\|\sin\theta_i$.
For unit vectors, \(\|u_i-b_i\|=2\sin(\theta_i/2)\), and therefore
\(\sin\theta_i\le \|u_i-b_i\|\). Using
\eqref{eq:thrust_direction_error} gives
\eqref{eq:att_error_bound}.
\end{proof}

\begin{remark}
The term \(e_i^{\rm GP}\) denotes only the translational
force-estimation error. Torque uncertainty is represented separately
by \(\Delta_{\tau,i}-\mu_{\tau,i}\) in the attitude dynamics and is
compensated in the nominal torque, but it does not enter the
second-order pairwise distance HOCBF margin. Thus,
Proposition~3.2 separates the translational learning error from the
error induced by underactuation and attitude tracking.
\end{remark}

\section{Robust Safety-Critical Control via HOCBF-QP}
\label{sec:robust_safety_filter}


\subsection{Learning uncertainty and robust HOCBF constraints}
\label{subsec:learning_uncertainty_hocbf}

For each agent $i$, let $\mu_{f,i}(x_i)$ denote the learned mean
prediction of $\Delta_{f,i}(x_i)$, and let $\mathcal X_i$ be a compact
operating domain containing its closed-loop trajectory. Assume the
learning model provides a high-probability confidence bound.
\begin{assumption}
\label{ass:gp_force_bound}
For each agent $i\in\mathcal V$, there exists a known function
$\bar\Delta_i:\mathcal X_i\rightarrow\mathbb R_{\geq0}$ and
a confidence level $\delta_i\in(0,1)$ such that, with probability at least $1-\delta_i$,
\begin{equation}
    \|\Delta_{f,i}(x_i)-\mu_{f,i}(x_i)\|
    \leq
    \bar\Delta_i(x_i),
    \qquad
    \forall x_i\in\mathcal X_i.
    \label{eq:force_confidence_bound}
\end{equation}
\end{assumption}

\medskip

For componentwise GP regression, we use
$\bar\Delta_i(x_i)=\sqrt{\beta_i}\|\sigma_i(x_i)\|$,
where $\sigma_i$ collects the posterior standard deviations
of the force components and $\beta_i>0$ determines the
confidence level under the adopted GP assumptions. The analysis below only requires
Assumption~\ref{ass:gp_force_bound}.

For each $(i,j)\in\mathcal E$, define the learned relative
acceleration
\begin{equation}
    \widehat a_{ij}
    :=
    \frac{f_i}{m_i}R_ie_3
    -
    \frac{f_j}{m_j}R_je_3
    +
    \frac{\mu_{f,i}(x_i)}{m_i}
    -
    \frac{\mu_{f,j}(x_j)}{m_j}.
    \label{eq:learned_relative_acceleration}
\end{equation}

Accordingly, the learned estimate of $\ddot h_{ij}$ is
\begin{equation}
    \widehat{\ddot h}_{ij}(x,f)
    :=
    2\|v_{ij}\|^2
    +
    2p_{ij}^{\top}\widehat a_{ij}.
    \label{eq:learned_hddot}
\end{equation}
The learned HOCBF expression is then
\begin{equation}
\begin{aligned}
    \widehat\varphi_{2,ij}(x,f)
    &:=
    \widehat{\ddot h}_{ij}(x,f)
    +
    \alpha_1'(h_{ij})\dot h_{ij} \\
    &\quad+
    \alpha_2\!\left(
        \dot h_{ij}+\alpha_1(h_{ij})
    \right).
\end{aligned}
\label{eq:learned_hocbf}
\end{equation}

The discrepancy between the true and learned HOCBF
expressions is bounded by
\begin{equation}
    \rho_{ij}(x)
    :=
    2\|p_{ij}\|
    \left(
        \frac{\bar\Delta_i(x_i)}{m_i}
        +
        \frac{\bar\Delta_j(x_j)}{m_j}
    \right).
    \label{eq:robustness_bound}
\end{equation}

\begin{lemma}\label{lem:robust_hocbf_deviation} 
Suppose that Assumption~\ref{ass:gp_force_bound} holds. Then, for
every $(i,j)\in\mathcal E$, the difference between the true and learned HOCBF expressions is independent of $f$ and satisfy
\begin{equation}
    \varphi_{2,ij}(x,f)
    \geq
    \widehat\varphi_{2,ij}(x,f)
    -
    \rho_{ij}(x),
    \label{eq:robust_hocbf_bound}
\end{equation}
for every admissible thrust vector $f$, with probability at least $1-\delta_i-\delta_j$.
\end{lemma}

\begin{proof}
The only difference between $\varphi_{2,ij}$ and
$\widehat\varphi_{2,ij}$ comes from the unknown force
in $\ddot h_{ij}$. From~\eqref{eq:h_ddot} and
\eqref{eq:learned_hddot},
\[
\begin{aligned}
    \ddot h_{ij}-\widehat{\ddot h}_{ij}
    &=
    2p_{ij}^{\top}
    \left(
        \frac{\Delta_{f,i}-\mu_{f,i}}{m_i}
        -
        \frac{\Delta_{f,j}-\mu_{f,j}}{m_j}
    \right).
\end{aligned}
\]
Therefore, by Cauchy-Schwarz and the triangle inequality and
Assumption~\ref{ass:gp_force_bound},
\[
\begin{aligned}
    \left|
        \ddot h_{ij}
        -
        \widehat{\ddot h}_{ij}
    \right|
    &\leq
    2\|p_{ij}\|
    \left(
        \frac{\|\Delta_{f,i}-\mu_{f,i}\|}{m_i}
        +
        \frac{\|\Delta_{f,j}-\mu_{f,j}\|}{m_j}
    \right)
    \\
    &\leq
    2\|p_{ij}\|
    \left(
        \frac{\bar\Delta_i(x_i)}{m_i}
        +
        \frac{\bar\Delta_j(x_j)}{m_j}
    \right)
    =
    \rho_{ij}(x).
\end{aligned}
\]
Since all other terms in $\varphi_{2,ij}$ and
$\widehat\varphi_{2,ij}$ are identical, the result follows.
The probability claim follows by a union bound over the learning events of agents $i$ and $j$.
\end{proof}

\subsection{Aggregated robust safety-critical QP}
\label{subsec:aggregated_qp}

The nominal thrust $f_i^{\rm nom}$ introduced in
Section~III implements the decentralized navigation force through
the current thrust direction $R_i e_3$. We now minimally modify
these nominal thrust magnitudes in order to enforce the robust
pairwise HOCBF conditions derived above.

Let $f:=\operatorname{col}(f_1,\ldots,f_N)$, $f^{\rm nom}:=
    \operatorname{col}(f_1^{\rm nom},\ldots,f_N^{\rm nom})$, and define the admissible thrust set
$\displaystyle{\mathcal F
    :=
    \prod_{i\in\mathcal V}
    [f_i^{\min},f_i^{\max}]}$. At each time $t$, the robust safety filter is defined by the
aggregated quadratic program
\begin{align}
    f^\star
    =
    \arg\min_{f\in\mathcal F}
    \quad&
    \frac{1}{2}
    \sum_{i\in\mathcal V}
    \bigl(f_i-f_i^{\rm nom}\bigr)^2
    \label{eq:aggregated_qp}
    \\
    \text{s.t.}\quad&
    \widehat\varphi_{2,ij}(x,f)
    \geq
    \rho_{ij}(x),
    \qquad
    (i,j)\in\mathcal E .
    \nonumber
\end{align}
The applied input of agent $i$ is then
\begin{equation}
    u_i^\star
    =
    \bigl(f_i^\star,\tau_i^{\rm nom}\bigr).
    \label{eq:safe_applied_input}
\end{equation} Since the objective is strictly convex and the constraints are
affine in $f$, whenever~\eqref{eq:aggregated_qp} is feasible
its minimizer is unique.

Problem~\eqref{eq:aggregated_qp} is solved simultaneously over the
thrust magnitudes of all agents, and therefore constitutes an
aggregated safety layer. The nominal navigation forces
$F_i^{\rm nav}$ and the corresponding thrust--attitude controllers
remain decentralized, whereas the HOCBF correction couples the
agents through the pairwise constraints associated with
$\mathcal E$. Each such constraint involves only the two thrust
magnitudes $f_i$ and $f_j$ corresponding to its endpoints.
\begin{remark}
Each constraint in~\eqref{eq:aggregated_qp} is affine in the thrust
magnitudes $f_i$ and $f_j$, since $\widehat\varphi_{2,ij}$ depends
linearly on $f_iR_ie_3$ and $f_jR_je_3$. The body torques do not
enter $\ddot h_{ij}$ directly; they affect safety through the
evolution of the thrust directions $R_ie_3$ and $R_je_3$.
Consequently, the safety filter modifies only the thrust magnitudes,
while the attitude torques remain those of the nominal geometric
controller in Section~\ref{sec:navigation_control}.
\end{remark}

\begin{remark}
The safety guarantee below is conditional on feasibility of
the robust HOCBF-QP. Feasibility may be limited by the thrust
bounds and, in particular, can deteriorate when $p_{ij}$ is
nearly orthogonal to both thrust directions $R_ie_3$ and
$R_je_3$, so that the second-order distance constraint has
poor instantaneous control authority. Higher-order or
torque-aware filters for such singular cases are left for
future work.
\end{remark}

\subsection{Probabilistic safety guarantee}
\label{subsec:probabilistic_safety}

Let $\mathcal C^{\rm H}
:=
\bigcap_{(i,j)\in\mathcal E}
\mathcal C_{ij}^{\rm H}$ be the aggregated safe set.

\begin{theorem}
\label{thm:probabilistic_safety}
Consider the multi-quadrotor system~\eqref{eq:kinematics}--%
\eqref{eq:rot_dyn} with pairwise safety functions
\eqref{eq:pairwise_safety_function}. Suppose that
Assumption~\ref{ass:gp_force_bound} holds for all agents on
the domain of the closed-loop trajectories and that
$x(0)\in\mathcal C^{\rm H}$. Assume that, for every
$t\geq0$, the aggregated QP~\eqref{eq:aggregated_qp} is
feasible and its optimizer $f^\star(x)$ is locally Lipschitz
in the state. Then $\mathcal C^{\rm H}$ is forward invariant
with probability at least $1-\sum_{i\in\mathcal V}\delta_i$. Consequently, every closed-loop trajectory initialized in
$\mathcal C^{\rm H}$ satisfies $\|p_i(t)-p_j(t)\|\geq d_{\min}$, $(i,j)\in\mathcal E$,
for all $t\geq0$, with the same probability.
\end{theorem}

\begin{proof}
For each $i\in\mathcal V$, let
\[
    \mathcal E_i
    :=
    \left\{
    \|\Delta_{f,i}(x_i)-\mu_{f,i}(x_i)\|
    \leq \bar\Delta_i(x_i),
    \ \forall x_i\in\mathcal X_i
    \right\}
\]
denote the confidence event of agent $i$, and define the joint
event $\mathcal E_{\rm GP}
    :=
    \bigcap_{i\in\mathcal V}\mathcal E_i$.
By Assumption~\ref{ass:gp_force_bound} and the union bound, $\mathbb P(\mathcal E_{\rm GP})
    \geq
    1-\sum_{i\in\mathcal V}\delta_i$.

We prove forward invariance conditioned on
$\mathcal E_{\rm GP}$. Fix an arbitrary edge
$(i,j)\in\mathcal E$. Since
$f^\star$ is feasible for~\eqref{eq:aggregated_qp},
\begin{equation}
    \widehat\varphi_{2,ij}(x,f^\star)
    \geq
    \rho_{ij}(x).
    \label{eq:proof_robust_constraint}
\end{equation}
On $\mathcal E_{\rm GP}$, the confidence bounds of both
agents $i$ and $j$ hold. Hence,
Lemma~\ref{lem:robust_hocbf_deviation} and
\eqref{eq:proof_robust_constraint} imply
\[
\begin{aligned}
    \varphi_{2,ij}(x,f^\star)
    &\geq
    \widehat\varphi_{2,ij}(x,f^\star)
    -
    \rho_{ij}(x)\geq0 .
\end{aligned}
\]

From the definitions of $\varphi_{2,ij}$ and $\varphi_{1,ij}$, we have 
$\varphi_{2,ij}
    =
    \dot\varphi_{1,ij}
    +
    \alpha_2(\varphi_{1,ij})$, and therefore
$\dot\varphi_{1,ij}
    \geq
    -\alpha_2(\varphi_{1,ij})$. Since $x(0)\in\mathcal C^{\rm H}$ implies
$\varphi_{1,ij}(0)\geq0$, the comparison principle yields $\varphi_{1,ij}(t)\geq0, t\geq0$. Similarly,
$\varphi_{1,ij}
    =
    \dot\varphi_{0,ij}
    +
    \alpha_1(\varphi_{0,ij})$, so that $\dot\varphi_{0,ij}
    \geq
    -\alpha_1(\varphi_{0,ij})$. Because $\varphi_{0,ij}(0)\geq0$, another application of the
comparison principle gives
$\varphi_{0,ij}(t)\geq0, t\geq0$. Thus $\mathcal C_{ij}^{\rm H}$ is forward invariant.

Since the edge $(i,j)$ was arbitrary and the same joint event
$\mathcal E_{\rm GP}$ guarantees the above inequalities for
all agents simultaneously,
$\mathcal C^{\rm H}
    =
    \bigcap_{(i,j)\in\mathcal E}
    \mathcal C_{ij}^{\rm H}$ is forward invariant on $\mathcal E_{\rm GP}$. Finally,
$\varphi_{0,ij}=h_{ij}$, and hence
$\mathcal C^{\rm H}\subseteq\mathcal C$. Therefore,
$h_{ij}(t)\geq0$ implies $\|p_i(t)-p_j(t)\|\geq d_{\min}$ for every $(i,j)\in\mathcal E$ and all $t\geq0$.
Combining this deterministic implication on
$\mathcal E_{\rm GP}$ with
$\mathbb P(\mathcal E_{\rm GP})
\geq1-\sum_i\delta_i$
proves the claim.
\end{proof}
\begin{corollary}
\label{cor:minimal_intervention}
Suppose that the hypotheses of
Theorem~\ref{thm:probabilistic_safety} hold. If, on an interval
$I\subset\mathbb R_{\geq0}$, the nominal thrust vector
$f^{\rm nom}(x)$ is feasible for
\eqref{eq:aggregated_qp}, then
$f^\star(t)=f^{\rm nom}(t)$, $t\in I$.
Consequently,
$u_i^\star(t)=u_i^{\rm nom}(t)$ for every $i\in\mathcal V$
and $t\in I$, and the closed-loop translational dynamics
coincide with the thrust--attitude implementation described in
Proposition~\ref{prop:implementation_error}.
\end{corollary}

\begin{proof}
If $f^{\rm nom}$ is feasible, the objective in
\eqref{eq:aggregated_qp} attains its minimum value zero at
$f=f^{\rm nom}$. Since the objective is strictly convex,
$f^\star=f^{\rm nom}$ is the unique minimizer. Equation
\eqref{eq:safe_applied_input} then gives
$u_i^\star=u_i^{\rm nom}$ for every agent, and the last claim
follows from Proposition~\ref{prop:implementation_error}.
\end{proof}

\begin{remark}
The navigation law of Section~\ref{sec:navigation_control} is decentralized, whereas the
safety filter~\eqref{eq:aggregated_qp} is solved jointly. Distributed
realizations of CBF-constrained multi-agent optimization have been
developed using auxiliary and consensus dynamics
\cite{TanDimarogonasDistributed,MestresDistributedCBF}.
In our setting, however, bounded scalar thrusts and multiple
simultaneous HOCBF constraints share inputs whose effectiveness
depends on $p_{ij}^{\top}R_i e_3$; we therefore retain the aggregated
QP. Recursively feasible distributed implementations and the possible
emergence of safety-filter-induced undesired equilibria
\cite{TanDimarogonasUndesired,MestresSafetyFilters} are left for
future work.
\end{remark}

\section{Simulation Results}
\label{sec:simulations}

Consider a six-UAV cinematography-inspired reconfiguration in which
the vehicles exchange camera viewpoints around a moving scene,
producing crossing nominal trajectories. Camera-heading variables are
used only for visualization, while safety is enforced through pairwise
position constraints. The translational dynamics are affected by an
unknown spatially varying wind disturbance learned online by each UAV.

The navigation force~\eqref{eq:nav_force} uses the DNF
\eqref{eq:decentralized_navigation_function} of~\cite{YagoNietoSimoesGiribetColombo2026DNFSE3Learning}, with $\kappa=1$,
and is mapped into thrust and attitude commands through the
underactuated tracking controller. Each UAV learns the translational disturbance from residual
samples
$y_i=m_i\dot v_i-f_iR_ie_3-m_ig$ using
$[p_i^\top,v_i^\top]^\top\in\mathbb R^6$ as input.
Three independent scalar GPs are used for the force components, with
kernel
$k(x,x')=1.0\,k_{\rm RBF}(x,x';\ell=2.0)+k_W(x,x')$,
noise level $2.5\times10^{-4}$, posterior standard-deviation
floor $2.2\times10^{-2}\,{\rm N}$ per force component,
$\beta=18$, update period $0.5\,{\rm s}$, a minimum of $18$ samples before fitting, and at most $300$ samples per UAV. All UAVs have mass \(m_i=1.30\,\mathrm{kg}\); thrust and torque
are bounded by \(0\leq f_i\leq28\,\mathrm{N}\) and
\(|\tau_{i,k}|\leq9\,\mathrm{N\,m}\), and the dynamics are integrated
with time step \(0.01\,\mathrm{s}\).

\begin{figure}[h!]
    \centering
    \includegraphics[width=\columnwidth]{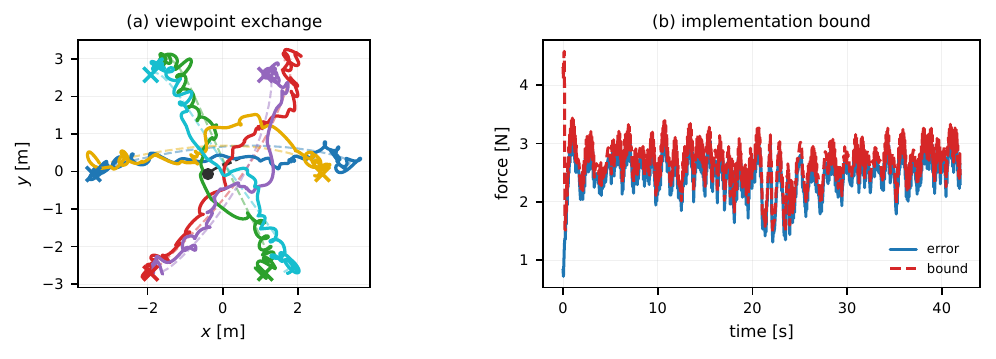}
    \caption{Baseline GP-aided thrust--attitude implementation.
(a) Viewpoint exchange. See video at: \url{https://youtu.be/XcRLlHjsYzQ}. (b) Total implementation error and
theoretical bound.}
    \label{fig:section3_baseline}
\end{figure}

Fig~\ref{fig:section3_baseline} validates the thrust--attitude
implementation of Section~III. The six UAVs complete the viewpoint
exchange while the total implementation error remains below its
theoretical bound throughout the maneuver, with maxima
\(3.2396\,\mathrm{N}\) and \(4.7892\,\mathrm{N}\), respectively.
The maximum residual in~\eqref{eq:implemented_nav_dynamics}
is \(4.68\times10^{-15}\,\mathrm{N}\), confirming consistency with
Proposition~\ref{prop:implementation_error}.


\begin{figure}[h!]
    \centering
    \includegraphics[width=\columnwidth]{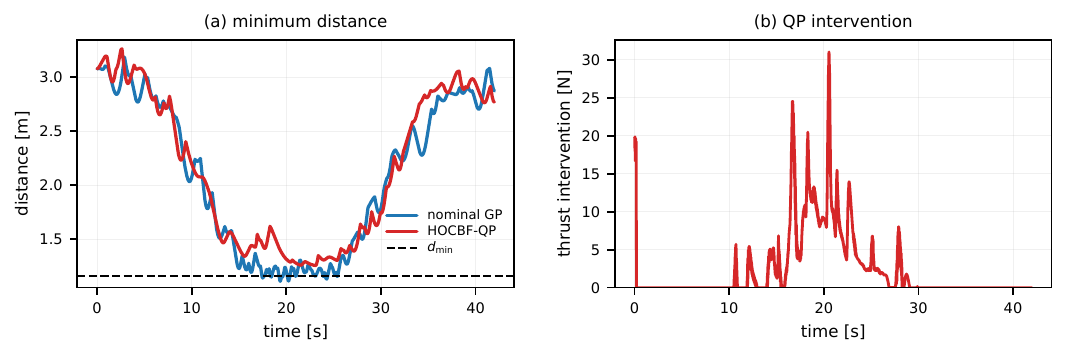}
    \caption{Robust HOCBF-QP safety correction.
(a) Minimum pairwise distance for the nominal and safe controllers. (b)  Team thrust correction $\|f^\star-f^{\rm nom}\|_2$.}
    \label{fig:section4_hocbf}
\end{figure}

We next use the same task and local GP models and activate the
aggregated robust HOCBF-QP of Section~IV. We set
\(d_{\min}=1.16\,\mathrm{m}\), \(\ell_1=1.6\), \(\ell_2=2.0\), and use
the full robustness bound \(\rho_{ij}\) in~\eqref{eq:robustness_bound}.
Fig~\ref{fig:section4_hocbf} shows that the nominal GP-compensated
controller violates the prescribed separation, reaching
\(1.1093\,\mathrm{m}\), whereas the HOCBF-QP maintains a minimum
distance of \(1.2540\,\mathrm{m}\). The QP is feasible throughout, with
\(100\%\) feasibility and solver success ratios and a maximum constraint
violation of \(2.68\times10^{-12}\). The mean and maximum values of
\(\|f^\star-f^{\rm nom}\|_2\) are \(2.3870\,\mathrm{N}\) and
\(31.0348\,\mathrm{N}\), respectively. The realized GP error remains within the adopted confidence bound
throughout the robust run, with a minimum coverage margin of
\(2.18\times10^{-2}\,\mathrm{N}\).

To isolate the role of the robustness term, we compare the aggregated
robust HOCBF-QP with a nonrobust version obtained by setting
\(\rho_{ij}=0\). Both controllers use the same nominal navigation law,
local GP compensation, and thrust--attitude implementation.
Fig~\ref{fig:robustness_ablation} shows that the nonrobust QP
maintains the prescribed distance, with a minimum of
\(1.1762\,\mathrm{m}\), but violates the GP-aware tightened certificate,
reaching a maximum deficit of \(6.9773\). In contrast, the robust QP
maintains a minimum distance of \(1.2540\,\mathrm{m}\) and keeps the
certificate deficit below \(2.68\times10^{-12}\), up to numerical
precision. Thus, the robustness margin is required to enforce the
GP-aware tightened HOCBF condition underlying
Theorem~\ref{thm:probabilistic_safety}.

\begin{figure}[!t]
    \centering
    \includegraphics[width=\columnwidth]{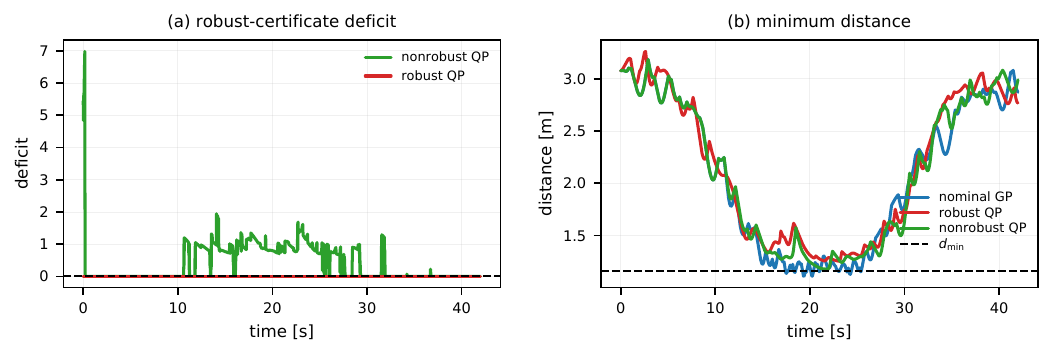}
    \caption{Effect of the GP robustness margin.
(a) Robust-certificate deficit for the nonrobust and robust HOCBF-QPs.
(b) Minimum pairwise distance for the nominal GP, nonrobust, and robust
controllers.}
    \label{fig:robustness_ablation}
\end{figure}

\section{Conclusions}
We developed a safety-critical control framework for quadrotor teams
combining decentralized navigation functions with an aggregated robust
HOCBF-QP safety layer under learned model uncertainty. The approach
implements the nominal navigation forces through geometric
thrust--attitude control, quantifies the resulting implementation error,
and incorporates GP force uncertainty into explicit pairwise robustness
margins. The resulting safety filter minimally modifies the nominal
thrusts and guarantees probabilistic collision avoidance whenever the
robust QP remains feasible. Future work will address distributed
realizations of the coupled safety QP under bounded thrust, as well as
torque-aware treatments of thrust-direction singularities and
experimental validation.

\bibliographystyle{IEEEtran}
\bibliography{autosam}

\end{document}